\pdfoutput=1
\documentclass[10pt,pra,aps,superscriptaddress,nofootinbib,twocolumn,longbibliography]{revtex4-1}
\usepackage{amsmath,bbm}
\usepackage{amsthm}
\usepackage{amsmath}
\usepackage{latexsym}
\usepackage{amssymb}
\usepackage{float}
\usepackage{graphicx}           
\usepackage{color}
\usepackage{xcolor}
\usepackage{mathpazo}
\usepackage{comment}
\usepackage{enumitem}
\usepackage{multirow}
\usepackage{setspace}
\usepackage{subcaption}
\usepackage[colorlinks=true,linkcolor=blue,citecolor=magenta,urlcolor=blue]{hyperref}
\usepackage{svg}

\newcommand{\be}{\begin{equation}}
\newcommand{\ee}{\end{equation}}
\newcommand{\bea}{\begin{eqnarray}}
\newcommand{\eea}{\end{eqnarray}}
\def\squareforqed{\hbox{\rlap{$\sqcap$}$\sqcup$}}
\def\qed{\ifmmode\squareforqed\else{\unskip\nobreak\hfil
\penalty50\hskip1em\null\nobreak\hfil\squareforqed
\parfillskip=0pt\finalhyphendemerits=0\endgraf}\fi}
\def\endenv{\ifmmode\;\else{\unskip\nobreak\hfil
\penalty50\hskip1em\null\nobreak\hfil\;
\parfillskip=0pt\finalhyphendemerits=0\endgraf}\fi}

\newcommand{\tr}{\text{Tr}}
\newcommand{\I}{\mathbbm{1}}

\newcommand{\ket}[1]{|#1\rangle}
\newcommand{\bra}[1]{\langle#1|}

\newcommand{\la}{\langle}
\newcommand{\ra}{\rangle}

\makeatletter
\newtheorem*{rep@theorem}{\rep@title}
\newcommand{\newreptheorem}[2]{%
\newenvironment{rep#1}[1]{%
 \def\rep@title{#2 \ref{##1}}%
 \begin{rep@theorem}}%
 {\end{rep@theorem}}}
\makeatother

\newtheorem{thm}{Theorem}
\newreptheorem{thm}{Theorem}
\newtheorem{lemma}{Lemma}

\newtheorem{example}{Example}

\begin{document}

\title{Single-shot antidistinguishability of unitary operations}

\author{Satyaki Manna}
\email{mannasatyaki@gmail.com}
\affiliation{Department of Physics, School of Basic Sciences, Indian Institute of Technology Bhubaneswar, Odisha 752050, India}
\affiliation{School of Physics, Indian Institute of Science Education and Research Thiruvananthapuram, Kerala 695551, India}
\author{Anandamay Das Bhowmik}
\email{ananda.adb@gmail.com}
\affiliation{School of Physics, Indian Institute of Science Education and Research Thiruvananthapuram, Kerala 695551, India}


\begin{abstract}
    The notion of antidistinguishability captures the possibility of ruling out certain alternatives in a quantum experiment without identifying the actual outcome. Although extensively studied for quantum states, the antidistinguishability of quantum channels remains largely unexplored. In this work, we investigate the single-shot antidistinguishability of unitary operations. We analyse two scenarios: antidistinguishability with single-system probes and with entangled probes. For sets of three unitaries, we first prove that all maximally entangled states are equivalent in their performance as probes. In the qubit case, we further establish that maximally entangled probes are always sufficient for perfect antidistinguishability: if a set of three unitaries acting on two-dimensional Hilbert space is antidistinguishable with either a single-system or non-maximally entangled probe, then it is also antidistinguishable with a maximally entangled one. However, in higher dimension,  this sufficiency fails. From \textit{dimension 3}, there exists a set of unitaries that are antidistinguishable with non-maximally entangled probe or single-system probe but not with maximally entangled probe. We also establish that union of two antidistinguishable sets of three unitaries acting on two-dimensional Hilbert space also forms a set of antidistinguishable unitaries. Lastly, we provide methods to construct antidistinguishable unitaries from non-antidistinguishable ones. 
\end{abstract}

\maketitle


\section{Introduction}
The concept of distinguishability of physical processes plays a crucial role in understanding quantum theory and, by extension, the physical world. Alongside distinguishability, there exists a weaker notion known as antidistinguishability. While distinguishability allows one to identify which process occurred from a set of known processes based on an outcome, antidistinguishability enables the negative identification of certain processes. In addition to distinguishability \cite{leifer2014,barrett,chaturvedi2021,Chaturvedi2020quantum,bhowmik2022}, the notion of antidistinguisahbility has significant implications for discussions on the reality of quantum states \cite{pbr,ray2024}. Beyond foundational insights, distinguishability and antidistinguishability find broad applications in quantum information and communication \cite{PhysRevResearch.6.043269,pandit,leifer}.

Extensive research on these tasks has focused primarily on quantum states, though the antidistinguishability of quantum states \cite{Caves02,Heinosaari_2018,Johnston2025tightbounds,PhysRevResearch.2.013326,bandyopadhyay,uola,yao2025} has received comparatively less attention than distinguishability \cite{Hellstrom,Chefles_2000,benett,walgate,watrous,vedral,JánosABergou_2007,Virmani_2001,ghosh,halder,ghosh_s}. Despite being fewer in number, there is a considerable body of work addressing the distinguishability of quantum channels \cite{Puchala,watrous2,sacchi2005,acin2001,Dariano,ji2006,duan2009,harrow2010,ziman2010,maximallyentangled,PhysRevA.111.022221,njp2021,globalvslocal}. However, the study of antidistinguishability for quantum channels remains nearly unexplored except some recent works \cite{PRXQuantum,ji2025barycentricbounds,ji2025conversebounds,PhysRevA.111.022221}. This disparity arises from the inherent complexity of the problem: antidistinguishability requires consideration of at least three different processes, as the notions of distinguishability and antidistinguishability coincide in the context of two processes. One of the earliest treatments of antidistinguishability for channels is presented in \cite{PhysRevA.111.022221}, where the channels under consideration are quantum measurements. In a similar vein, this work focusses on the antidistinguishability of unitary operations, since dealing with general quantum channels is often highly intricate.

Unitary operators represent one of the most fundamental classes of quantum channels and have a significant importance in quantum communication \cite{densecoding,Brub}, information processing \cite{science.1090790}, error correction \cite{You_2012}, cryptography \cite{ekert}, and computation\cite{HUANG2022127863}. While there exists a body of research on the discrimination of unitary operators \cite{acin2001,maximallyentangled} in various contexts, the problem of antidiscrimination remains unexplored. The present study aims to provide a foundational perspective on the antidistinguishability of unitaries, analogous to the well-studied discrimination problem. 

This work formulates the problem of antidistinguishing a priori known sets of unitary operations sampled from an equal probability distribution. Before addressing the main problem, a brief overview of the antidistinguishability of quantum states is provided. The notion of antidistinguishability for unitary operations is then defined within the single-shot scenario, demonstrating that the antidistinguishability of unitaries ultimately reduces to the antidistinguishability of the corresponding evolved states for an optimally chosen initial probe state. Based on the nature of the initial probe, scenarios are categorized into two types: antidistinguishability with a single-system probe and antidistinguishability with an entangled-system probe. The study primarily is directed on evaluating the antidistinguishability of three unitaries.
Subsequently, it is shown that all maximally entangled states perform equivalently in the antidistinguishability task for any three unitaries. The analysis then progresses to ranking different probes in this context. The results reveal that for any three unitaries acting on $\mathbbm{C}^2$, if they are antidistinguishable using a non-maximally entangled probe or a single-system probe, they are also antidistinguishable using a maximally entangled state. This implies that the maximally entangled state serves as the optimal probe for antidistinguishing three unitaries acting on  two-dimensional Hilbert space. However, this property does not hold when the unitary dimension increases to three. In higher dimensions, there exists a set of three unitaries that are antidistinguishable using single-system or non-maximally entangled probes but not with maximally entangled probes. After that, we move into the results regarding the features of a set of antidistinguishable unitaries. We showed that the union of two antidistinguishable sets of three unitaries acting on $\mathbbm{C}^2$ also forms a set of antidistinguishable unitaries. Additionally, some methods are proposed for constructing antidistinguishable sets from those that are not antidistinguishable. These findings draw direct parallels to previously established results on the antidistinguishability of quantum states.

The structure of the paper is as follows: the subsequent section provides an overview of the antidistinguishability of quantum states. Following this, the single-shot antidistinguishing task for unitary operations is formulated. Section \ref{results} presents the main results. The conclusion summarizes the key findings and discusses open problems and possible directions for future research.

\section{antidistinguishability of quantum states}
Assume, we are given $n$ previously known quantum states $\{\rho_k\}_{k=1}^n$ sampled from the probability distribution $\{q_k\}_{k=1}^n$.
Antidistinguishability of $n$ quantum states $\{\rho_k\}_{k=1}^n$ is a linear function \cite{PhysRevA.111.022221} which is defined as,
\bea  
\mathcal{A}[\{\rho_k\}_k,\{q_k\}_k] &= & \max_{\{M_a\}_a}\Bigg\{ \sum_{k,a} q_k  p(a\neq k|\rho_k,M_a)\Bigg\}.
\eea
 We are taking optimization over the measurement $\{M_a\}_{a=1}^n$ such that the above quantity attains its maximum value. As $\sum_a (p(a\neq k|\rho_k,M)+p(a=k|\rho_k,M))=1$ and $\sum_k q_k =1$, the above expression becomes,
\bea\label{pA}
\mathcal{A}[\{\rho_k\}_k,\{q_k\}_k] 
&=& 1 - \min_{\{M_k\}_k}\left\{\sum_{k}q_k \ \tr(\rho_k M_k)\right\} .
\eea 

The necessary and sufficient conditions \cite{Caves02} for perfect antidistinguishability of three pairwise non-orthogonal pure quantum states, i.e.,
\be \label{ADSpsi123}
\mathcal{A}[\{\ket{\psi_1},\ket{\psi_2},\ket{\psi_3}\}] = 1,
\ee 
are as follows:
\begin{subequations}\label{condAS}
\be
    x_1 + x_2 + x_3 < 1 
\ee
\be 
(x_1+x_2+x_3-1)^2 \geqslant 4x_1x_2x_3 ,
\ee 
\text{  where $x_1=|\la \psi_1|\psi_2\ra|^2, x_2=|\la \psi_1|\psi_3\ra|^2, x_3=|\la \psi_2|\psi_3\ra|^2$.}\\
\end{subequations}

If we increase the number of the states in the set, ref. \cite{Johnston2025tightbounds} gives some necessary and sufficient (not simultaneously) conditions for the antidistinguishability and non-antidistinguishability of the given set depending on the modulus of pairwise inner products of the states.  One useful sufficient condition is as follows:
For any set of states $\{\ket{\psi_i}\}_{i=1}^n$, if
\be\label{scond}
|\la\psi_i|\psi_j\ra|\leqslant\frac{1}{\sqrt{2}}\sqrt{\frac{n-2}{n-1}},
\ee
for all $1\leqslant i\neq j\leqslant n$, then the set is antidistinguishable.

\section{antidistinguishability of unitary operations}
Unitary operator $(U)$ is a linear operator $U:H\rightarrow H$ on a Hilbert space $H$ that satisfies $U^\dagger U=UU^\dagger=\I$.

We consider a priori known set of $r$ unitaries $\{U_x\}_x$ acting on a $d$-dimensional quantum state and $x\in\{1,\cdots,r\}$. We assume that all the unitaries are sampled from a probability distribution $\{p_x\}_x$, i.e., $p_x>0$ and $\sum_x p_x=1$. In this paper, we consider the unitaries which are not distinguishable as distinguishable unitaries are trivially antidistinguishable. To antidistinguish the unitaries, the unitary device is fed with a known quantum state, single or entangled and the device carries out one of $r$ unitary operations. After this process, the device gives an evolved state as the output. Therefore, one can perform any measurement on the evolved state and this measurement can be optimized such that the antidistinguishability of these evolved states will be maximum. Let us describe this optimal measurement by a set of POVM elements $\{N_{b}\}_b$, where $b\in\{1,\cdots,r\}$ and $\sum_{b=1}^r N_{b}=\I$. The protocol is successful in antidistinguishing the unitaries if $b$ is not equal to $x$. Any classical post-processing of outcome $b$ can be included in the measurement $\{N_{b}\}_b$. Depending on the initial probing state, we can formulate two situations which are described in detail in the next subsection.
\subsection{With single system probe}
At first, the unitary device is given a single quantum state $\rho$. After the device applying any of the unitaries from the set $\{U_x\}_{x=1}^r$, the output state will be $U_x\rho U_x^\dagger$. After performing the measurement $\{N_b\}_{b=1}^r$, the antidistinguishability of the set of unitaries becomes the antidistinguishability of the set of evolved states. So antidistinguishability in this scenario, denoted by $\mathcal{A}_S$, is defined as,
\bea\label{A_S}
\mathcal{A}_S\left[\{U_x\}_x,\{p_x\}_x\right]&=& \max_\rho \sum_x p_xp(b\neq x|x)\nonumber\\
&=& 1-\max_\rho \sum_x p_xp(b=x|x)\nonumber\\
&=& 1-\min_{\rho,N_b}\sum_x p_x\tr\left(U_x\rho U_x^\dagger N_{b=x}\right)\nonumber\\
&=& \max_{\rho} \mathcal{A}\left[\left\{U_x\rho U_x^\dagger\right\}_x,\{p_x\}_x\right].\nonumber
\eea
$\mathcal{A}\left[\{\rho_k\}_k,\{q_k\}_k\right]$ denotes the antidistinguishability of the set of states $\{\rho_k\}_k$, sampled from the probability distribution $\{q_k\}_k$. As of now, we are interested to antidistinguish three unitary operations $U_1, U_2, U_3$. If a pure state $\rho=\ket{\psi}\bra{\psi}$ acts as the probe, we need to compute the square of inner product terms $\left|\bra{\psi} U_i^\dagger U_j\ket{\psi}\right|^2$, where $i,j\in\{1,2,3\}$ and $i\neq j$, to find the antidistinguishability of three unitaries by employing \eqref{condAS}. As $U_i^\dagger U_j$ is an unitary, we can write the spectral decomposition of this operator as, $U_i^\dagger U_j=\sum_{l=1}^d e^{\mathbbm{i}\theta_{l_{ij}}}\ket{\psi_{l_{ij}}}\bra{\psi_{l_{ij}}}$, where $e^{\mathbbm{i}\theta_{l_{ij}}}$ are the eigenvalues of $U_i^\dagger U_j$ and $\ket{\psi_{l_{ij}}}$ is the eigenvector corresponding to $l$'th eigenvalue. Therefore $\ket{\psi}$ can be decomposed into the linear combination of the eigenstates $\ket{\psi_{l_{ij}}}$, i.e., $\ket{\psi}=\sum_{l=1}^d\alpha_{l_{ij}}\ket{\psi_{l_{ij}}} $. So,
\bea\label{As_conv}
\left|\bra{\psi} U_i^\dagger U_j\ket{\psi}\right|^2&=& \left|\sum_{l=1}^d]\left|\alpha_{l_{ij}}\right|^2 e^{\mathbbm{i}\theta_{l_{ij}}}\right|^2\nonumber\\
&=& \left|con\{e^{\mathbbm{i}\theta_{l_{ij}}}\}\right|^2,
\eea
where $con\{e^{\mathbbm{i}\theta_{l_{ij}}}\}$ denotes the set of complex numbers that can be written as the convex combinations of $\{e^{\mathbbm{i}\theta_{l_{ij}}}\}$.
\subsection{With entangled probe}
Now, the unitary device is fed with a $d\otimes d'$ entangled state $\rho_{AB}$ and the device applies one of the unitaries from the set $\{U_x\}_{x=1}^r$ on the part A of the entangled state. The evolved state will be $(U_x\otimes\I)\rho_{AB}(U_x^\dagger\otimes\I)$. Then one can perform the optimal measurement of dimension $dd'$. Similarly, the antidistinguishability of the unitaries in this scenario, denoted as $\mathcal{A}_E$, can be written as,
\bea\label{A_E}
&&\mathcal{A}_E\left[\left\{U_x\right\}_x,\{p_x\}_x\right]\nonumber\\
&=& \max_{\rho_{AB}} \sum_x p_xp(b\neq x|x)\nonumber\\
&=& 1-\min_{\rho_{AB},N_b} \sum_x p_x \tr\left[(U_x\otimes\I)\rho_{AB}(U_x^\dagger\otimes\I)N_{b=x}\right]\nonumber\\
&=& \max_{\rho_{AB}} \mathcal{A}\left[\left\{(U_x\otimes\I)\rho_{AB}(U_x^\dagger\otimes\I)\right\}_x,\left\{p_x\right\}_x\right].
\eea
This expression suggests that if the input state is product state, this expression will reduce to \eqref{A_S}, i.e., the distinguishability with single system scenario. So, in the context of probing state, we use single state and product state alternatively. 
\begin{lemma}\label{suff_prob}
    For antidistinguishability in entanglement assisted scenario, the sufficient initial entangled state is of $d\otimes d$ dimension for $d$ dimensional unitaries.
\end{lemma}
\begin{proof}
     We can take the best possible entangled probing state $\ket{\psi}_{AB}\in\mathcal{C}^d\otimes\mathcal{C}^{d'}$. By Schmidt decomposition, we can always write the state $\ket{\psi}_{AB}=\sum_{l=1}^{d} C_l\ket{\eta_l}\ket{\chi_l}$, where $d'\geqslant d$. From the structure of the Schmidt decomposition, we see that the extra $(d'-d)$ basis states are redundant. 
    
In the other case, when $d'<d$, the input state can be written as $\ket{\psi'}_{AB}=\sum_{l=1}^{d'} C_l\ket{\eta_l}\ket{\chi_l}$. Such a state is not suitable for our purpose because it does not have support on all basis states $\ket{\eta_l}$, and hence cannot probe the full $d$-dimensional space.
 So the sufficiency of $d\otimes d$ dimensional entangled probe is proven hereby.

      For more details, one can check the ref. \cite{maximallyentangled}.
\end{proof}
For more detailed analysis, we need to compute the square of inner product terms $\left|\bra{\psi_{AB}} (U_i^\dagger\otimes\I) (U_j\otimes\I)\ket{\psi_{AB}}\right|^2$, where $i,j\in\{1,2,3\}$ and $\ket{\psi_{AB}} = \sum_{w=1}^{d} C_w\ket{\eta_w}_A\ket{\chi_w}_B$.
\bea\label{ent_con}
&&\left|\bra{\psi_{AB}} (U_i^\dagger\otimes\I) (U_j\otimes\I)\ket{\psi_{AB}}\right|^2\nonumber\\
&=& \left|\sum_{w=1}^d \left|C_w\right|^2\bra{\eta_w}U_i^\dagger U_j\ket{\eta_w}\right|^2\nonumber\\
&=& \left|\sum_{w=1}^d \left|C_w\right|^2\sum_{l=1}^d\left|\beta^w_{l_{ij}}\right|^2 e^{\mathbbm{i}\theta_{l_{ij}}}\right|^2\nonumber\\
&=&\left|\sum_{l=1}^d\sum_{w=1}^d \left|C_w\right|^2 \left|\beta^w_{l_{ij}}\right|^2 e^{\mathbbm{i}\theta_{l_{ij}}}\right|^2\nonumber\\
&=& \left|con\{e^{\mathbbm{i}\theta_{l_{ij}}}\}\right|^2.
\eea
One can check $\sum_{l=1}^d\sum_{w=1}^d |C_w|^2 |\beta^w_{l_{ij}}|^2 =1$, that means these coefficients indeed make a convex combinations of the eigenvalues $e^{\mathbbm{i}\theta_{l_{ij}}}$. Third line comes from the fact that $U_i^\dagger U_j=\sum_{l=1}^d e^{\mathbbm{i}\theta_{l_{ij}}}\ket{\psi_{l_{ij}}}\bra{\psi_{l_{ij}}}$ and $\ket{\eta_w}=\sum_{l=1}^d \beta^w_{l_{ij}}\ket{\psi_{l_{ij}}}$.
\section{Results}\label{results}
In this section, we present most of the results regarding the antidistinguishability of three unitaries and all the unitaries are sampled from equal probability distribution. Our first theorem reflects the usefulness of maximally entangled state as the probe in any dimension. 
\begin{thm}\label{thm1}
    All the maximally entangled states are equivalent for antidistinguishability of three unitaries.
\end{thm}
\begin{proof}
    Let us take the maximally entangled state $\ket{\Phi^+}=\frac{1}{\sqrt{d}}\sum_{k=1}^d\ket{kk}$ as the initial input state. We need to compute $\left|\bra{\Phi^+} U_i^\dagger U_j\ket{\Phi^+}\right|^2$, where $i,j\in\{1,2,3\}$ and $i\neq j$. So,
    \bea\label{max_entg}
    \left|\bra{\Phi^+} U_i^\dagger U_j\ket{\Phi^+}\right|^2 &=& \frac{1}{d^2}\left|\sum_{k=1}^d\bra{k}U_i^\dagger U_j\ket{k}\right|^2\nonumber\\
    &=& \frac{1}{d^2}\left|\tr(U_i^\dagger U_j)\right|^2\nonumber\\
    &=& \frac{1}{d^2}\left|\sum_{l=1}^d e^{\mathbbm{i}\theta_{l_{ij}}}\right|^2.
    \eea
    As trace is basis independent property, it does not depend on the initial maximally entangled state and that gives our desired result.
\end{proof}
Now we are interested between the hierarchy of the probes. For the simplest case, we take any three unitaries acting on $\mathbbm{C}^2$ and show that the maximally entangled state is the best probe for perfect antidistinguishability. 

\begin{thm}\label{thm2}
   For the perfect antidistinguishability of three unitaries acting on  two dimensional Hilbert space,  the optimum probe is the maximally entangled probe. 
\end{thm}
\begin{proof}
        We prove this theorem by proving the following two statements:\\
       $(i)$ If any three unitaries acting on  two dimensional Hilbert space are antidistinguishable with non-maximally entangled probe, they will be antidistinguishable with maximally entangled probe.

       $(ii)$ If any three unitaries acting on  two dimensional Hilbert space are antidistinguishable with single system probe, they will be antidistinguishable with maximally entangled probe.\\
       It is needless to say that if these above statements are true, we can infer the Theorem \ref{thm2}.

       Suppose three unitaries $A_1,A_2,A_3$, which acting on $\mathbbm{C}^2$, are antidistinguishable with non-maximally entangled probe $\ket{\tau}$. Let us denote the following quantities:\\
       \bea\label{th8}
       g_1 &=& \left|\bra{\tau}A_1^\dagger A_2\ket{\tau}\right|^2\nonumber\\
       &=& \left|con\{e^{\mathbbm{i}\Theta^1_{12}}, e^{\mathbbm{i}\Theta^2_{12}}\}\right|^2,\nonumber\\
       g_2 &=& \left|\bra{\tau}A_2^\dagger A_3\ket{\tau}\right|^2\nonumber\\
       &=& \left|con\{e^{\mathbbm{i}\Theta^1_{23}}, e^{\mathbbm{i}\Theta^2_{23}}\}\right|^2,\nonumber\\
       g_3 &=& \left|\bra{\tau}A_3^\dagger A_1\ket{\tau}\right|^2\nonumber\\
       &=& \left|con\{e^{\mathbbm{i}\Theta^1_{31}}, e^{\mathbbm{i}\Theta^2_{31}}\}\right|^2.\nonumber\\
       \eea
   $e^{\mathbbm{i}\Theta^1_{ij}}$ and $e^{\mathbbm{i}\Theta^2_{ij}}$ are the eigenvalues of $A_i^\dagger A_j$. For non-maximally entangled state, the convex combinations of eigenvalues are of unequal probability weights, whereas maximally entangled probe gives convex combinations of equal weights. Now we will calculate the quantities of \eqref{th8} for both the non-maximally and maximally entangled probe. For non-maximally entangled probe, using \eqref{ent_con}, we find,
   \bea
   g_{1_{NM}} &=& \left|t e^{\mathbbm{i}\Theta^1_{12}}+(1-t) e^{\mathbbm{i}\Theta^2_{12}}\right|^2\nonumber\\
   &=& t^2+(1-t)^2+ 2 t(1-t)\cos{(\Theta^1_{12}-\Theta^2_{12})}\nonumber\\
   &=& \frac{1+\cos{(\Theta^1_{12}-\Theta^2_{12})}}{2}\nonumber\\
   &&+2(t-\frac12)^2(1-\cos{(\Theta^1_{12}-\Theta^2_{12})}).
   \eea
    The suffix $'NM'$ denotes "non-maximally entangled probe" and the suffix $'M'$ is used for "maximally entangled probe". Similarly, one can calculate $g_{2_{NM}}$ and $g_{3_{NM}}$. Now, for maximally entangled probe,
   \bea
    g_{1_{M}} &=&  \frac14\left|e^{\mathbbm{i}\Theta^1_{12}}+ e^{\mathbbm{i}\Theta^2_{12}}\right|^2\nonumber\\
    &=&\frac{1+\cos{(\Theta^1_{12}-\Theta^2_{12})}}{2}.
   \eea
   It is easy to notice that, for any values of $\Theta^1_{12}$ and $\Theta^2_{12}$, $g_{1_{M}}\leqslant g_{1_{NM}}$ as $(t-\frac12)^2\geqslant 0$ and $(1-\cos{(\Theta^1_{12}-\Theta^2_{12})})\geqslant 0$. $g_{1_{M}}= g_{1_{NM}}$ when $t=\frac12$. Similarly, $g_{2_{M}}\leqslant g_{2_{NM}}$ and $g_{3_{M}}\leqslant g_{3_{NM}}$. As per the statement of the theorem, the unitaries are antidistinguishable with non-maximally entangled state. From \eqref{condAS}, we can write,
   \bea\label{thm8}
    && g_{1_{NM}} + g_{2_{NM}} + g_{3_{NM}} < 1, \nonumber\\
    && (g_{1_{NM}}+g_{2_{NM}}+g_{3_{NM}}-1)^2 \geqslant 4g_{1_{NM}}g_{2_{NM}}g_{3_{NM}}.\nonumber\\
   \eea
   We know,
   $g_{1_{M}} + g_{2_{M}} + g_{3_{M}} \leqslant g_{1_{NM}} + g_{2_{NM}} + g_{3_{NM}}$, which is less than $1$ from \eqref{thm8}. So first condition is satisfied for maximally entangled probe. Now for the second condition, $(g_{1_{M}} + g_{2_{M}} + g_{3_{M}}-1)^2 \geqslant (g_{1_{NM}} + g_{2_{NM}} + g_{3_{NM}}-1)^2\geqslant 4g_{1_{NM}}g_{2_{NM}}g_{3_{NM}}\geqslant 4g_{1_{M}}g_{2_{M}}g_{3_{M}}$. So $g_{1_{M}}, g_{2_{M}}$ and $g_{3_{M}}$ satisfied the necessary and sufficient conditions of antidistinguishability. That achieves our claim of first statement.

If three unitaries $A_1,A_2,A_3$ are antidistinguishable with single system $\ket{\Delta}$, we can write from \eqref{As_conv},
      \bea
      g_{1_{S}}&=&\left|\bra{\Delta}A_1^\dagger A_2\ket{\Delta}\right|^2\nonumber\\
      &=& \left|t e^{\mathbbm{i}\Omega^1_{12}}+(1-t) e^{\mathbbm{i}\Omega^2_{12}}\right|^2,
      \eea
      where $t\in(0,1)$ and $e^{\mathbbm{i}\Omega^1_{12}}, e^{\mathbbm{i}\Omega^2_{12}}$ are the eigenvalues of $A_1^\dagger A_2$. We use the suffix 'S' to denote single system. Maximally entangled probe yields convex combination of the eigenvalues with equal weights.
      \bea
    g_{1_{M}} &=& \frac14\left|e^{\mathbbm{i}\Omega^1_{12}}+ e^{\mathbbm{i}\Omega^2_{12}}\right|^2.
      \eea
      We already proved that $g_{1_{M}} \leqslant g_{1_{NM}}=g_{1_{S}}$. A similar argument holds for the remaining terms, exactly as we did earlier. As $g_{1_{S}},g_{2_{S}},g_{3_{S}}$ satisfy the necessary and sufficient conditions of \eqref{condAS}, $g_{1_{M}}, g_{2_{M}}, g_{3_{M}}$ also satisfy the same conditions. This completes the proof.
\end{proof}
But if we increase the dimension of the unitaries to $3$, it is very easy to construct three unitaries which are antidistinguishable with non-maximally entangled state or single system but not antidistinguishable with maximally entangled state. 

\begin{widetext}
\begin{thm}\label{thr_3}
Consider the following three unitaries:\\
\bea\label{Vu}
V_1 &=& \ket{\omega_1}\!\bra{1} + \ket{\omega_2}\!\bra{2} + \ket{\omega_3}\!\bra{3} +\cdots+ \ket{\omega_d}\!\bra{d},\nonumber\\
V_2 &=& \left(\frac12\ket{\omega_1} +\frac{\sqrt{3}}{2}\ket{ \omega_2}\right)\!\bra{1} + \left(-\frac{\sqrt{3}}{2}\ket{\omega_1} +\frac12\ket{ \omega_2}\right)\!\bra{2} + \ket{\omega_3}\!\bra{3}+\cdots+ \ket{\omega_d}\!\bra{d},\nonumber\\
V_3 &=& \left(\frac12\ket{\omega_1} +\frac{\sqrt{3}}{2}\ket{ \omega_3}\right)\!\bra{1} + \ket{\omega_2}\!\bra{2} + \left(-\frac{\sqrt{3}}{2}\ket{\omega_1} +\frac12\ket{ \omega_3}\right)\!\bra{3} +\cdots+ \ket{\omega_d}\!\bra{d}.\nonumber\\
\eea
These unitaries are antidistinguishable with both the non-maximally entangled probe and the single system probe but not antidistinguishable with maximally entangled probe for $d\geqslant 3$.
\end{thm}
\end{widetext}
\begin{proof}
    For non-maximally entangled state, we select the probe $(\sqrt{p}\ket{11}+\sqrt{1-p}\ket{22})$ with $p\in(\frac{14-2\sqrt{21}}{7},1]$, $p=1$ corresponds to single system. With this probe, the evolved states will be $\{\sqrt{p}\ket{\omega_1}\ket{1}+\sqrt{1-p}\ket{\omega_2}\ket{2},$ $ \sqrt{p}\left(\frac12\ket{\omega_1} + \frac{\sqrt{3}}{2}\ket{ \omega_2}\right)\ket{1} + \sqrt{1-p}\left(-\frac{\sqrt{3}}{2}\ket{\omega_1} + \frac12\ket{ \omega_2}\right)\ket{2},$ $\sqrt{p}\left(\frac12\ket{\omega_1} + \frac{\sqrt{3}}{2}\ket{ \omega_3}\right)\ket{1}+\sqrt{1-p}\ket{\omega_2}\ket{2}\}$. These states will be antidistinguishable if they satisfy the necessary and sufficient conditions given at \eqref{condAS}. 

    Using those conditions, we get the following two inequalities.
    \bea
&(i)& \frac14+(1-\frac{p}{2})^2+(\frac12-\frac{p}{4})^2 < 1;\\
&(ii)& \left(\frac14+(1-\frac{p}{2})^2+(\frac12-\frac{p}{4})^2-1\right)^2\nonumber\\
&&\geqslant 4\frac14(1-\frac{p}{2})^2(\frac12-\frac{p}{4})^2
    \eea
 Solving first inequality, we find $p\in[\frac{10-2\sqrt{15}}{5},1]$ and from second inequality, we get $p\in(\frac{14-2\sqrt{21}}{7},1]$. So the necessary range is $p\in(\frac{14-2\sqrt{21}}{7},1]$.

 Inspired from the previous calculation, We take the single system probe $\ket{1}$. The evolved states after unitaries being operated are $\{\ket{\omega_1}, \frac12\ket{\omega_1} +\frac{\sqrt{3}}{2}\ket{ \omega_2}, \frac12\ket{\omega_1} +\frac{\sqrt{3}}{2}\ket{ \omega_3}\}$. These three states are antidistinguishable as these states satisfy the conditions \eqref{condAS}. 
 
 So three unitaries described at \eqref{Vu} are antidistinguishable with non-maximally entangled probe as well as single system probe.
 
  For maximally entangled probe, we have to check the following quantities from \eqref{max_entg}:\\
   \bea
   y_1 &=&\frac{1}{d^2}\left|\tr(V_1^\dagger V_2)\right|^2= (1-\frac1d)^2,\nonumber\\
   y_2 &=&\frac{1}{d^2}\left|\tr(V_2^\dagger V_3)\right|^2= (1-\frac{7}{4d})^2,\nonumber\\
   y_3 &=&\frac{1}{d^2}\left|\tr(V_3^\dagger V_1)\right|^2= (1-\frac1d)^2.\nonumber\\
   \eea
     For $d\geqslant 3$, $y_1+y_2+y_3 > 1$. From  \eqref{condAS}, we can infer $V_1,V_2,V_3$ are not antidistinguishable with maximally entangled state. 
\end{proof}
Now we extend this hierarchy checking between single system and entangled system.

At this point, we diverted the next results to the characteristics of a set of antidistinguishable unitaries. Our next theorem depicts the closure property of antidistinguishable sets of three unitaries acting on  two-dimensional Hilbert space.
\begin{thm}
    If $\mathbbm{S}_1$ and $\mathbbm{S}_2$ are the two sets, each consists of three antidistinguishable qubit unitaries, the unitaries of the set $\mathbbm{S}_1\bigcup\mathbbm{S}_2$ are also antidistinguishable.
\end{thm}
\begin{proof}
      From Theorems \ref{thm2}, we can assert that maximally entangled state is the sufficient probe to antidistinguish three unitaries acting on $\mathbbm{C}^2$. Theorem \ref{thm1} tells that any maximally entangled state is equivalent for the task. So that makes possible to find a common probe, i.e., any maximally entangled probe which can antidistinguish both the sets $\mathbbm{S}_1$ and $\mathbbm{S}_2$. After using this probe, we can denote the corresponding two sets of evolved states as $\mathbbm{S'}_1$ and $\mathbbm{S'}_2$, which are the sets of antidistinguishable states. From proposition $2$ of \cite{Heinosaari_2018}, we know that $\mathbbm{S'}_1\bigcup\mathbbm{S'}_2$ is also a set of antidistinguishable states. That facilitates that the unitaries of the set $\mathbbm{S}_1\bigcup\mathbbm{S}_2$ are antidistinguishable.
\end{proof}
Our next results tell the procedure of transforming a set of non-antidistinguishable unitaries to antidistinguishable unitaries.
\begin{thm}\label{thm7}
    Any finite set of qubit unitaries are either antidistinguishable or the set can be made antidistinguishable by adding one unitary.
\end{thm}
\begin{proof}
    For any single qubit state $\ket{\chi}$ as a probing state, the antidistinguishability of any finite set of unitaries, i.e., $\{\mathcal{U}_1,\cdots,\mathcal{U}_n\}$ eventually reduces to the antidistinguishability of the set of evolved states $\mathbbm{S} = \{\mathcal{U}_1\ket{\chi},\cdots,\mathcal{U}_n\ket{\chi}\}.$ Proposition $7$ of \cite{Heinosaari_2018} tells that any set of qubit pure states are either antidistinguishable or it can be made antidistinguishable by adding one pure qubit state. Similarly we can say either $\mathbbm{S}$ is a set of antidistinguishable unitaries or we can add one more unitary to make the set antidistinguishable without changing the probing state.
\end{proof}

While Theorem \ref{thm7} ensures that any finite set of qubit unitaries can be extended to an antidistinguishable set, below we show that, having tensor product of unitaries does not, in general, preserve their non-antidistinguishability.

\begin{example}
    There exist two not antidistinguishable qubit unitaries $W_1 = \ket{0}\!\bra{0}+\ket{1}\!\bra{1}, W_2 = \ket{+}\!\bra{0}-\ket{-}\!\bra{1}$ such that the set $\{W_i\otimes W_j\}_{i,j=1}^2$ are antidistinguishable.
\end{example}

Eigenvalues of $W_1^\dagger W_2$ are $(\frac{1+\mathbbm{i}}{\sqrt{2}}, \frac{1-\mathbbm{i}}{\sqrt{2}})$. So none of the convex combination of these two eigenvalues can give zero, that means $W_1$ and $W_2$ are not distinguishable \cite{maximallyentangled}, which implies they are not antidistinguishable. \\
Now we take the following $2\otimes 2$ four unitaries $\{W_1\otimes W_1, W_1\otimes W_2, W_2\otimes W_1, W_2\otimes W_2 \}$. With the product probe $\ket{00}$, the four evolved states are $\{\ket{00},\ket{0+},\ket{+0},\ket{++}\}$. The ref.  \cite{pbr} showed that these four states are antidistinguishable with an entangled measurement.

Since indistinguishability and non-antidistinguishability are equivalent notions for a pair of unitaries, it is instructive to identify a set of three unitaries that are individually non-antidistinguishable yet whose tensor products constitute an antidistinguishable set.
\begin{example}
    There exist three not antidistinguishable qubit unitaries $Q_1 = \ket{0}\!\bra{0}+\ket{1}\!\bra{1}, Q_2 = \ket{\eta}\!\bra{0}-\ket{\eta^\perp}\!\bra{1}$ and $Q_3 = \ket{\zeta}\!\bra{0}-\ket{\zeta^\perp}\!\bra{1}$ such that the set $\{Q_i\otimes Q_j\}_{i,j=1}^3$ are antidistinguishable, where    $\ket{\eta}=\cos(5\pi/18)\ket{0}+\sin(5\pi/18)\ket{1}, \ket{\eta^\perp}=\sin(5\pi/18)\ket{0}-\cos(5\pi/18)\ket{1}, \ket{\zeta}=\cos(19\pi/60)\ket{0}+e^{\mathbbm{i}2\pi/3}\sin(19\pi/60)\ket{1} $ and $\ket{\zeta^\perp}=\sin(19\pi/60)\ket{0}-e^{\mathbbm{i}2\pi/3}\cos(19\pi/60)\ket{1}$.
\end{example}

    We have already established that any maximally entangled state serves as a sufficient probe for verifying the antidistinguishability of three unitaries acting on $\mathbbm{C}^2$ (Theorems~\ref{thm1} and \ref{thm2}). Using the Bell state  $
\ket{\phi^+} = \frac{1}{\sqrt{2}}(\ket{00} + \ket{11})
$
as the probe, the three unitaries \(\{Q_i\}_{i=1}^3\) generate the evolved states 
$
\Big\{\ket{\phi^+},\ \ket{\phi'} = \tfrac{1}{\sqrt{2}}(\ket{\eta}\ket{0} - \ket{\eta^\perp}\ket{1}),\ 
\ket{\Bar{\phi}} = \tfrac{1}{\sqrt{2}}(\ket{\zeta}\ket{0} - \ket{\zeta^\perp}\ket{1})\Big\}.
$
A direct calculation shows that 
$
|\langle\phi^+|\phi'\rangle|^2 + |\langle\phi'|\Bar{\phi}\rangle|^2 + |\langle\phi^+|\Bar{\phi}\rangle|^2 
< 4\, |\langle\phi^+|\phi'\rangle|^2\, |\langle\phi'|\Bar{\phi}\rangle|^2\, |\langle\phi^+|\Bar{\phi}\rangle|^2,
$
which is precisely the opposite of the second condition in~\eqref{condAS}. This indicates that the three states are not antidistinguishable, and consequently, the corresponding unitaries are not antidistinguishable either.

Next, we consider the extended set of unitaries \(\{Q_i \otimes Q_j\}_{i,j=1}^3\). Choosing the probe state \(\ket{00}\), we obtain the nine evolved states:
$
\{\ket{00},\ \ket{0 \eta },\ \ket{ 0\zeta},\ \ket{\eta 0},\ \ket{\eta \eta},\ \ket{\eta \zeta},\ \ket{\zeta 0},\ \ket{\zeta \eta},\ \ket{\zeta \zeta}\}.
$

 These nine states satisfy the sufficient condition of antidistinguishability given by \eqref{scond}. This completes the proof.

\section{Conclusion}
In this work, we initiated a systematic study of the single-shot antidistinguishability of unitary operations. By formulating the task in analogy with the well-studied case of quantum states, we demonstrated that the problem ultimately reduces to the antidistinguishability of appropriately evolved states for optimally chosen probe states. Our analysis revealed several key structural insights. For three unitaries, all the maximally entangled probes work equivalently in antidistinguishaing task. For three unitaries acting on  two-dimensional Hilbert space, we established that maximally entangled probes are always sufficient.  However, we also showed that this hierarchy breaks down in higher dimensions, where there exist examples of unitaries that are antidistinguishable only with single-system and non-maximally entangled probes but not with maximally entangled probes.
Beyond these structural results, we proved closure properties of antidistinguishable sets of three unitaries acting on $\mathbbm{C}^2$ and identified procedures for constructing antidistinguishable sets from those that are not. These findings parallel and extend known results in the antidistinguishability of quantum states, highlighting a deeper connection between the two frameworks.

Our results leave several open questions. Most notably, the characterization of antidistinguishable sets of higher-dimensional unitaries remains unexplored, and the role of entanglement in such settings warrants further investigation. Another intriguing direction is the exploration of multi-shot scenarios and the potential resource-theoretic formulation of antidistinguishability for quantum channels. We hope that this work lays the foundation for a broader understanding of the operational and structural aspects of antidistinguishability in quantum theory.
\subsection*{Acknowledgment}
The authors thank Debashis Saha for fruitful discussion. ADB acknowledges support from STARS (STARS/STARS-2/2023-0809), Govt. of India.
\bibliography{ref} 

@article{Caves02,
  title = {Conditions for compatibility of quantum-state assignments},
  author = {Caves, Carlton M. and Fuchs, Christopher A. and Schack, R\"udiger},
  journal = {Phys. Rev. A},
  volume = {66},
  issue = {6},
  pages = {062111},
  numpages = {11},
  year = {2002},
  month = {Dec},
  publisher = {American Physical Society},
  doi = {10.1103/PhysRevA.66.062111},
  url = {https://link.aps.org/doi/10.1103/PhysRevA.66.062111}
}

@article{Johnston2025tightbounds,
  doi = {10.22331/q-2025-02-04-1622},
  url = {https://doi.org/10.22331/q-2025-02-04-1622},
  title = {Tight bounds for antidistinguishability and circulant sets of pure quantum states},
  author = {Johnston, Nathaniel and Russo, Vincent and Sikora, Jamie},
  journal = {{Quantum}},
  issn = {2521-327X},
  publisher = {{Verein zur F{\"{o}}rderung des Open Access Publizierens in den Quantenwissenschaften}},
  volume = {9},
  pages = {1622},
  month = feb,
  year = {2025}
}

@article{Heinosaari_2018,
doi = {10.1088/1751-8121/aad1fc},
url = {https://dx.doi.org/10.1088/1751-8121/aad1fc},
year = {2018},
month = {jul},
publisher = {IOP Publishing},
journal ={Journal of Physics A:
Mathematical and
Theoretical},
volume = {51},
number = {36},
pages = {365303},
author = {Heinosaari, Teiko and Kerppo, Oskari},
title = {Antidistinguishability of pure quantum states}
}

@article{PhysRevResearch.2.013326,
  title = {Simple communication complexity separation from quantum state antidistinguishability},
  author = {Havl\'{\i}\ifmmode \check{c}\else \v{c}\fi{}ek, Vojt\ifmmode \check{e}\else \v{e}\fi{}ch and Barrett, Jonathan},
  journal = {Phys. Rev. Res.},
  volume = {2},
  issue = {1},
  pages = {013326},
  numpages = {8},
  year = {2020},
  month = {Mar},
  publisher = {American Physical Society},
  doi = {10.1103/PhysRevResearch.2.013326},
  url = {https://link.aps.org/doi/10.1103/PhysRevResearch.2.013326}
}

@article{maximallyentangled,
  title = {Limitation of maximally entangled probes for single-shot distinguishability of unitaries},
  author = {Manna, Satyaki and Das Bhowmik, Anandamay and Saha, Debashis},
  journal = {Phys. Rev. A},
  volume = {112},
  issue = {4},
  pages = {042215},
  numpages = {8},
  year = {2025},
  month = {Oct},
  publisher = {American Physical Society},
  doi = {10.1103/bsxv-q9x7},
  url = {https://link.aps.org/doi/10.1103/bsxv-q9x7}
}

@article{pbr,
    author = {Pusey, M. and Barrett, J. and Rudolph, T. },
    title = {On the reality of the quantum state},
    journal = {Nature Phys},
    year = {2012},
    doi={10.1038/nphys2309}
}

@article{PhysRevA.111.022221,
  title = {Single-shot distinguishability and antidistinguishability of quantum measurements},
  author = {Manna, Satyaki and Suresh, Sneha and Kachhawaha, Manan Singh and Saha, Debashis},
  journal = {Phys. Rev. A},
  volume = {111},
  issue = {2},
  pages = {022221},
  numpages = {18},
  year = {2025},
  month = {Feb},
  publisher = {American Physical Society},
  doi = {10.1103/PhysRevA.111.022221},
  url = {https://link.aps.org/doi/10.1103/PhysRevA.111.022221}
}

@misc{ray2024,
      title={No epistemic model can explain anti-distinguishability of quantum mixed preparations}, 
      author={Sagnik Ray and Visweshwaran R and Debashis Saha},
      year={2024},
      eprint={2401.17980},
      archivePrefix={arXiv},
      primaryClass={quant-ph},
      url={https://arxiv.org/abs/2401.17980}, 
}

@article{leifer,
  title = {Noncontextuality inequalities from antidistinguishability},
  author = {Leifer, Matthew and Duarte, Cristhiano},
  journal = {Phys. Rev. A},
  volume = {101},
  issue = {6},
  pages = {062113},
  numpages = {11},
  year = {2020},
  month = {Jun},
  publisher = {American Physical Society},
  doi = {10.1103/PhysRevA.101.062113},
  url = {https://link.aps.org/doi/10.1103/PhysRevA.101.062113}
}

@article{PhysRevResearch.6.043269,
  title = {Unbounded quantum advantage in communication complexity measured by distinguishability},
  author = {Manna, Satyaki and Chaturvedi, Anubhav and Saha, Debashis},
  journal = {Phys. Rev. Res.},
  volume = {6},
  issue = {4},
  pages = {043269},
  numpages = {13},
  year = {2024},
  month = {Dec},
  publisher = {American Physical Society},
  doi = {10.1103/PhysRevResearch.6.043269},
  url = {https://link.aps.org/doi/10.1103/PhysRevResearch.6.043269}
}

@misc{pandit,
      title={Limits of Classical correlations and Quantum advantages under (Anti-)Distinguishability constraints in Multipartite Communication}, 
      author={Ankush Pandit and Soumyabrata Hazra and Satyaki Manna and Anubhav Chaturvedi and Debashis Saha},
      year={2025},
      eprint={2506.07699},
      archivePrefix={arXiv},
      primaryClass={quant-ph},
      url={https://arxiv.org/abs/2506.07699}, 
}

@article{bandyopadhyay,
  title = {Conclusive exclusion of quantum states},
  author = {Bandyopadhyay, Somshubhro and Jain, Rahul and Oppenheim, Jonathan and Perry, Christopher},
  journal = {Phys. Rev. A},
  volume = {89},
  issue = {2},
  pages = {022336},
  numpages = {13},
  year = {2014},
  month = {Feb},
  publisher = {American Physical Society},
  doi = {10.1103/PhysRevA.89.022336},
  url = {https://link.aps.org/doi/10.1103/PhysRevA.89.022336}
}

@article{uola,
  title = {All Quantum Resources Provide an Advantage in Exclusion Tasks},
  author = {Uola, Roope and Bullock, Tom and Kraft, Tristan and Pellonp\"a\"a, Juha-Pekka and Brunner, Nicolas},
  journal = {Phys. Rev. Lett.},
  volume = {125},
  issue = {11},
  pages = {110402},
  numpages = {6},
  year = {2020},
  month = {Sep},
  publisher = {American Physical Society},
  doi = {10.1103/PhysRevLett.125.110402},
  url = {https://link.aps.org/doi/10.1103/PhysRevLett.125.110402}
}

@misc{yao2025,
      title={Conclusive exclusion of quantum states with group action}, 
      author={Hongshun Yao and Xin Wang},
      year={2025},
      eprint={2503.04605},
      archivePrefix={arXiv},
      primaryClass={quant-ph},
      url={https://arxiv.org/abs/2503.04605}, 
}

@book{Hellstrom,
    author ={C. W. Helstrom} ,
    title = { Quantum Detection and Estimation Theory},
    journal ={J Stat Phys 1, 231–252 (1969)} ,
    publisher={Academic Press, New York},
    year = {1969},
    doi={https://doi.org/10.1007/BF01007479}
}

@article{Chefles_2000,
   title={Quantum state discrimination},
   volume={41},
   ISSN={1366-5812},
   url={http://dx.doi.org/10.1080/00107510010002599},
   DOI={10.1080/00107510010002599},
   number={6},
   journal={Contemporary Physics},
   publisher={Informa UK Limited},
   author={Chefles, Anthony},
   year={2000},
   month=nov, pages={401–424} }

@article{benett,
  title = {Quantum nonlocality without entanglement},
  author = {Bennett, Charles H. and DiVincenzo, David P. and Fuchs, Christopher A. and Mor, Tal and Rains, Eric and Shor, Peter W. and Smolin, John A. and Wootters, William K.},
  journal = {Phys. Rev. A},
  volume = {59},
  issue = {2},
  pages = {1070--1091},
  numpages = {0},
  year = {1999},
  month = {Feb},
  publisher = {American Physical Society},
  doi = {10.1103/PhysRevA.59.1070},
  url = {https://link.aps.org/doi/10.1103/PhysRevA.59.1070}
}

@article{vedral,
  title = {Local Distinguishability of Multipartite Orthogonal Quantum States},
  author = {Walgate, Jonathan and Short, Anthony J. and Hardy, Lucien and Vedral, Vlatko},
  journal = {Phys. Rev. Lett.},
  volume = {85},
  issue = {23},
  pages = {4972--4975},
  numpages = {0},
  year = {2000},
  month = {Dec},
  publisher = {American Physical Society},
  doi = {10.1103/PhysRevLett.85.4972},
  url = {https://link.aps.org/doi/10.1103/PhysRevLett.85.4972}
}

@article{Virmani_2001,
   title={Optimal local discrimination of two multipartite pure states},
   volume={288},
   ISSN={0375-9601},
   url={http://dx.doi.org/10.1016/S0375-9601(01)00484-4},
   DOI={10.1016/s0375-9601(01)00484-4},
   number={2},
   journal={Physics Letters A},
   publisher={Elsevier BV},
   author={Virmani, S and Sacchi, M.F and Plenio, M.B and Markham, D},
   year={2001},
   month=sep, pages={62–68} }

@article{walgate,
  title = {Nonlocality, Asymmetry, and Distinguishing Bipartite States},
  author = {Walgate, Jonathan and Hardy, Lucien},
  journal = {Phys. Rev. Lett.},
  volume = {89},
  issue = {14},
  pages = {147901},
  numpages = {4},
  year = {2002},
  month = {Sep},
  publisher = {American Physical Society},
  doi = {10.1103/PhysRevLett.89.147901},
  url = {https://link.aps.org/doi/10.1103/PhysRevLett.89.147901}
}

@article{watrous,
  title = {Bipartite Subspaces Having No Bases Distinguishable by Local Operations and Classical Communication},
  author = {Watrous, John},
  journal = {Phys. Rev. Lett.},
  volume = {95},
  issue = {8},
  pages = {080505},
  numpages = {4},
  year = {2005},
  month = {Aug},
  publisher = {American Physical Society},
  doi = {10.1103/PhysRevLett.95.080505},
  url = {https://link.aps.org/doi/10.1103/PhysRevLett.95.080505}
}

@article{JánosABergou_2007,
doi = {10.1088/1742-6596/84/1/012001},
url = {https://dx.doi.org/10.1088/1742-6596/84/1/012001},
year = {2007},
month = {oct},
volume = {84},
number = {1},
pages = {012001},
author = {János A Bergou},
title = {Quantum state discrimination and selected applications},
journal = {Journal of Physics: Conference Series}
 }

@article{Puchala,
  title = {Strategies for optimal single-shot discrimination of quantum measurements},
  author = {Pucha\l{}a, Zbigniew and Pawela, \L{}ukasz and Krawiec, Aleksandra and Kukulski, Ryszard},
  journal = {Phys. Rev. A},
  volume = {98},
  issue = {4},
  pages = {042103},
  numpages = {11},
  year = {2018},
  month = {Oct},
  publisher = {American Physical Society},
  doi = {10.1103/PhysRevA.98.042103},
  url = {https://link.aps.org/doi/10.1103/PhysRevA.98.042103}
}

@article{acin2001,
  title = {Statistical Distinguishability between Unitary Operations},
  author = {Ac\'{\i}n, A.},
  journal = {Phys. Rev. Lett.},
  volume = {87},
  issue = {17},
  pages = {177901},
  numpages = {4},
  year = {2001},
  month = {Oct},
  publisher = {American Physical Society},
  doi = {10.1103/PhysRevLett.87.177901},
  url = {https://link.aps.org/doi/10.1103/PhysRevLett.87.177901}
}

@article{sacchi2005,
  title = {Optimal discrimination of quantum operations},
  author = {Sacchi, Massimiliano F.},
  journal = {Phys. Rev. A},
  volume = {71},
  issue = {6},
  pages = {062340},
  numpages = {4},
  year = {2005},
  month = {Jun},
  publisher = {American Physical Society},
  doi = {10.1103/PhysRevA.71.062340},
  url = {https://link.aps.org/doi/10.1103/PhysRevA.71.062340}
}

@article{duan2009,
  title = {Perfect Distinguishability of Quantum Operations},
  author = {Duan, Runyao and Feng, Yuan and Ying, Mingsheng},
  journal = {Phys. Rev. Lett.},
  volume = {103},
  issue = {21},
  pages = {210501},
  numpages = {4},
  year = {2009},
  month = {Nov},
  publisher = {American Physical Society},
  doi = {10.1103/PhysRevLett.103.210501},
  url = {https://link.aps.org/doi/10.1103/PhysRevLett.103.210501}
}

@article{harrow2010,
  title = {Adaptive versus nonadaptive strategies for quantum channel discrimination},
  author = {Harrow, Aram W. and Hassidim, Avinatan and Leung, Debbie W. and Watrous, John},
  journal = {Phys. Rev. A},
  volume = {81},
  issue = {3},
  pages = {032339},
  numpages = {6},
  year = {2010},
  month = {Mar},
  publisher = {American Physical Society},
  doi = {10.1103/PhysRevA.81.032339},
  url = {https://link.aps.org/doi/10.1103/PhysRevA.81.032339}
}

@article{ziman2010,
author = { Ziman M\'{a}rio  and  Michal   Sedl\'{a}k },
title = {Single-shot discrimination of quantum unitary processes},
journal = {Journal of Modern Optics},
volume = {57},
number = {3},
pages = {253-259},
year  = {2010},
publisher = {Taylor & Francis},
doi = {10.1080/09500340903349963},
URL = { https://doi.org/10.1080/09500340903349963}
}

@article{ji2006,
  title = {Identification and Distance Measures of Measurement Apparatus},
  author = {Ji, Zhengfeng and Feng, Yuan and Duan, Runyao and Ying, Mingsheng},
  journal = {Phys. Rev. Lett.},
  volume = {96},
  issue = {20},
  pages = {200401},
  numpages = {4},
  year = {2006},
  month = {May},
  publisher = {American Physical Society},
  doi = {10.1103/PhysRevLett.96.200401},
  url = {https://link.aps.org/doi/10.1103/PhysRevLett.96.200401}
}

@article{Dariano,
	doi = {10.1088/0305-4470/38/26/010},
	url = {https://doi.org/10.1088%2F0305-4470%2F38%2F26%2F010},
	journal={J. Phys. A: Math. Gen.},
	year = 2005,
	month = {jun},
	publisher = {{IOP} Publishing},
	volume = {38},
	number = {26},
	pages = {5979-5991},
	author = {Giacomo Mauro D'Ariano and Paoloplacido Lo Presti and Paolo Perinotti},
}

@article{watrous2,
  title = {All Entangled States are Useful for Channel Discrimination},
  author = {Piani, Marco and Watrous, John},
  journal = {Phys. Rev. Lett.},
  volume = {102},
  issue = {25},
  pages = {250501},
  numpages = {4},
  year = {2009},
  month = {Jun},
  publisher = {American Physical Society},
  doi = {10.1103/PhysRevLett.102.250501},
  url = {https://link.aps.org/doi/10.1103/PhysRevLett.102.250501}
}

@article{densecoding,
  title = {Communication via one- and two-particle operators on Einstein-Podolsky-Rosen states},
  author = {Bennett, Charles H. and Wiesner, Stephen J.},
  journal = {Phys. Rev. Lett.},
  volume = {69},
  issue = {20},
  pages = {2881--2884},
  numpages = {0},
  year = {1992},
  month = {Nov},
  publisher = {American Physical Society},
  doi = {10.1103/PhysRevLett.69.2881},
  url = {https://link.aps.org/doi/10.1103/PhysRevLett.69.2881}
}

@article{Brub,
  title = {Distributed Quantum Dense Coding},
  author = {Bru\ss{}, D. and D'Ariano, G. M. and Lewenstein, M. and Macchiavello, C. and Sen(De), A. and Sen, U.},
  journal = {Phys. Rev. Lett.},
  volume = {93},
  issue = {21},
  pages = {210501},
  numpages = {4},
  year = {2004},
  month = {Nov},
  publisher = {American Physical Society},
  doi = {10.1103/PhysRevLett.93.210501},
  url = {https://link.aps.org/doi/10.1103/PhysRevLett.93.210501}
}

@article{ekert,
  title = {Quantum cryptography based on Bell's theorem},
  author = {Ekert, Artur K.},
  journal = {Phys. Rev. Lett.},
  volume = {67},
  issue = {6},
  pages = {661--663},
  numpages = {0},
  year = {1991},
  month = {Aug},
  publisher = {American Physical Society},
  doi = {10.1103/PhysRevLett.67.661},
  url = {https://link.aps.org/doi/10.1103/PhysRevLett.67.661}
}

@article{HUANG2022127863,
title = {Query complexity of unitary operation discrimination},
journal = {Physica A: Statistical Mechanics and its Applications},
volume = {604},
pages = {127863},
year = {2022},
issn = {0378-4371},
doi = {https://doi.org/10.1016/j.physa.2022.127863},
url = {https://www.sciencedirect.com/science/article/pii/S0378437122005581},
author = {Xiaowei Huang and Lvzhou Li}
}

@article{
science.1090790,
author = {Joseph Emerson  and Yaakov S. Weinstein  and Marcos Saraceno  and Seth Lloyd  and David G. Cory },
title = {Pseudo-Random Unitary Operators for Quantum Information Processing},
journal = {Science},
volume = {302},
number = {5653},
pages = {2098-2100},
year = {2003},
doi = {10.1126/science.1090790},
URL = {https://www.science.org/doi/abs/10.1126/science.1090790},
eprint = {https://www.science.org/doi/pdf/10.1126/science.1090790}
}

@article{You_2012,
doi = {10.1088/0253-6102/58/3/09},
url = {https://doi.org/10.1088/0253-6102/58/3/09},
year = {2012},
month = {sep},
publisher = {},
volume = {58},
number = {3},
pages = {377},
author = {You Bo and  Xu Ke and Wu Xiao-Hua},
title = {Unitary Application of the Quantum Error Correction Codes},
journal = {Communications in Theoretical Physics}
}

@article{ghosh,
  title = {Distinguishability of maximally entangled states},
  author = {Ghosh, Sibasish and Kar, Guruprasad and Roy, Anirban and Sarkar, Debasis},
  journal = {Phys. Rev. A},
  volume = {70},
  issue = {2},
  pages = {022304},
  numpages = {7},
  year = {2004},
  month = {Aug},
  publisher = {American Physical Society},
  doi = {10.1103/PhysRevA.70.022304},
  url = {https://link.aps.org/doi/10.1103/PhysRevA.70.022304}
}

@article{halder,
  title = {Strong Quantum Nonlocality without Entanglement},
  author = {Halder, Saronath and Banik, Manik and Agrawal, Sristy and Bandyopadhyay, Somshubhro},
  journal = {Phys. Rev. Lett.},
  volume = {122},
  issue = {4},
  pages = {040403},
  numpages = {7},
  year = {2019},
  month = {Feb},
  publisher = {American Physical Society},
  doi = {10.1103/PhysRevLett.122.040403},
  url = {https://link.aps.org/doi/10.1103/PhysRevLett.122.040403}
}

@article{ghosh_s,
  title = {Activating strong nonlocality from local sets: An elimination paradigm},
  author = {Ghosh, Subhendu B. and Gupta, Tathagata and V., Ardra A. and Das Bhowmik, Anandamay and Saha, Sutapa and Guha, Tamal and Mukherjee, Amit},
  journal = {Phys. Rev. A},
  volume = {106},
  issue = {1},
  pages = {L010202},
  numpages = {7},
  year = {2022},
  month = {Jul},
  publisher = {American Physical Society},
  doi = {10.1103/PhysRevA.106.L010202},
  url = {https://link.aps.org/doi/10.1103/PhysRevA.106.L010202}
}

@article{njp2021,
doi = {10.1088/1367-2630/abecaf},
url = {https://dx.doi.org/10.1088/1367-2630/abecaf},
year = {2021},
month = {apr},
publisher = {IOP Publishing},
volume = {23},
number = {4},
pages = {043021},
author = {Chandan Datta and Tanmoy Biswas and Debashis Saha and Remigiusz Augusiak},
title = {Perfect discrimination of quantum measurements using entangled systems},
journal = {New Journal of Physics},
}

@misc{chaturvedi2021,
      title={Quantum description of reality is empirically incomplete}, 
      author={Anubhav Chaturvedi and Marcin Pawłowski and Debashis Saha},
      year={2021},
      eprint={2110.13124},
      archivePrefix={arXiv},
      primaryClass={quant-ph},
      url={https://arxiv.org/abs/2110.13124}, 
}

@article{Chaturvedi2020quantum,
  doi = {10.22331/q-2020-10-21-345},
  url = {https://doi.org/10.22331/q-2020-10-21-345},
  title = {Quantum prescriptions are more ontologically distinct than they are operationally distinguishable},
  author = {Chaturvedi, Anubhav and Saha, Debashis},
  journal = {{Quantum}},
  issn = {2521-327X},
  publisher = {{Verein zur F{\"{o}}rderung des Open Access Publizierens in den Quantenwissenschaften}},
  volume = {4},
  pages = {345},
  month = oct,
  year = {2020}
}

@misc{bhowmik2022,
      title={On the Interpretation of Quantum Indistinguishability : a No-Go Theorem}, 
      author={Anandamay Das Bhowmik and Preeti Parashar},
      year={2022},
      eprint={2204.09736},
      archivePrefix={arXiv},
      primaryClass={quant-ph},
      url={https://arxiv.org/abs/2204.09736}, 
}

@article{leifer2014,
  title = {$\ensuremath{\psi}$-Epistemic Models are Exponentially Bad at Explaining the Distinguishability of Quantum States},
  author = {Leifer, M. S.},
  journal = {Phys. Rev. Lett.},
  volume = {112},
  issue = {16},
  pages = {160404},
  numpages = {4},
  year = {2014},
  month = {Apr},
  publisher = {American Physical Society},
  doi = {10.1103/PhysRevLett.112.160404},
  url = {https://link.aps.org/doi/10.1103/PhysRevLett.112.160404}
}

@article{barrett,
  title = {No $\ensuremath{\psi}$-Epistemic Model Can Fully Explain the Indistinguishability of Quantum States},
  author = {Barrett, Jonathan and Cavalcanti, Eric G. and Lal, Raymond and Maroney, Owen J. E.},
  journal = {Phys. Rev. Lett.},
  volume = {112},
  issue = {25},
  pages = {250403},
  numpages = {6},
  year = {2014},
  month = {Jun},
  publisher = {American Physical Society},
  doi = {10.1103/PhysRevLett.112.250403},
  url = {https://link.aps.org/doi/10.1103/PhysRevLett.112.250403}
}

@misc{globalvslocal,
      title={Global vs. Local Discrimination of Locally Implementable Multipartite Unitaries}, 
      author={Satyaki Manna and Sneha Suresh and Anandamay Das Bhowmik and Debashis Saha},
      year={2025},
      eprint={2509.10430},
      archivePrefix={arXiv},
      primaryClass={quant-ph},
      url={https://arxiv.org/abs/2509.10430}, 
}

@article{PRXQuantum,
  title = {Exact Quantum Sensing Limits for Bosonic Dephasing Channels},
  author = {Huang, Zixin and Lami, Ludovico and Wilde, Mark M.},
  journal = {PRX Quantum},
  volume = {5},
  issue = {2},
  pages = {020354},
  numpages = {21},
  year = {2024},
  month = {Jun},
  publisher = {American Physical Society},
  doi = {10.1103/PRXQuantum.5.020354},
  url = {https://link.aps.org/doi/10.1103/PRXQuantum.5.020354}
}

@misc{ji2025barycentricbounds,
      title={Barycentric bounds on the error exponents of quantum hypothesis exclusion}, 
      author={Kaiyuan Ji and Hemant K. Mishra and Milán Mosonyi and Mark M. Wilde},
      year={2025},
      eprint={2407.13728},
      archivePrefix={arXiv},
      primaryClass={quant-ph},
      url={https://arxiv.org/abs/2407.13728}, 
}

@misc{ji2025conversebounds,
      title={Converse bounds for quantum hypothesis exclusion: A divergence-radius approach}, 
      author={Kaiyuan Ji and Hemant K. Mishra and Milán Mosonyi and Mark M. Wilde},
      year={2025},
      eprint={2501.09712},
      archivePrefix={arXiv},
      primaryClass={quant-ph},
      url={https://arxiv.org/abs/2501.09712}, 
}
\end{document}